\documentclass[10pt,letterpaper]{article}
\usepackage{acl2015,indentfirst,times,secdot}
\usepackage{xspace,empheq,fancybox,amsmath,bbm,amssymb,epsfig,subfig,syntonly,times,amsthm,graphicx} \usepackage{psfrag,color,bm,array}
\usepackage{url,cite,footnote,syntonly,algpseudocode}
\usepackage{verbatim,slashbox}
\usepackage[linesnumbered,ruled,vlined]{algorithm2e}

\usepackage{verbatim,multirow,diagbox,enumitem,tabu,acl2015}
\usepackage[T1]{fontenc}
\usepackage{setspace}
\usepackage[justification=justified,font={small,sf}]{caption}
\usepackage[table,xcdraw]{xcolor}
\usepackage{adjustbox,balance}

\usepackage[colorinlistoftodos]{todonotes}

\def\BibTeX{{\rm B\kern-.05em{\sc i\kern-.025em b}\kern-.08em
    T\kern-.1667em\lower.7ex\hbox{E}\kern-.125emX}}

\newtheorem{remark}{\bfseries Remark}
\input{mysymbol.sty}
\graphicspath{{/}}

\begin{document}
    
\title{Dynamic Response Recovery Using Ambient Synchrophasor Data:\\ A Synthetic Texas Interconnection Case Study}

\author{ \normalfont Shaohui Liu, Hao Zhu \\
 Department of ECE \\
 The University of Texas at Austin \\
 {\underline{ \{shaohui.liu, haozhu\}@utexas.edu}} \\ 
 \And
 \normalfont Vassilis Kekatos \\
 Bradley Department of ECE \\
 Virginia Tech \\
 {\underline{ kekatos@vt.edu}} \\
}

 \maketitle


\begin{abstract}
Wide-area dynamic studies are of paramount importance to ensure the stability and reliability of power grids. This paper puts forth a comprehensive framework for inferring the dynamic responses in the small-signal regime using ubiquitous fast-rate ambient data collected during normal grid operations. We have shown that the impulse response between any pair of locations can be recovered in a model-free fashion by cross-correlating angle and power flow data streams collected only at these two locations, going beyond previous work based on frequency data only. The result has been established via model-based analysis of linearized second-order swing dynamics under certain conditions. Numerical validations demonstrate its applicability  to realistic power system models including nonlinear, higher-order dynamics. In particular, the case study using  synthetic PMU data on a synthetic Texas Interconnection (TI) system strongly corroborates the benefit of using angle PMU data over frequency data for real-world power system dynamic modeling.
\end{abstract}


\section{Introduction}\label{sec:intro}

Power system dynamic studies are critical for achieving stability and security in control center operations \cite[Ch.~1]{kundur1994power}. Poorly damped or forced oscillations can notoriously affect the operation of  interconnected systems; see \cite{kosterev1999model,blackout2005,adams2012sub_synchronous,nerc_odessa_report2021,nerc_report2021}.
Thus, enhancing the modeling and analysis of dynamic responses is of imperative needs.

Recently, high-rate synchrophasor data have provided unprecedented visibility of grid transients. Using synchrophasor data, data-driven approaches have been developed for estimating dynamical models; e.g., \cite{zhang2016dependency,chavan2016identification}. However, most of them are limited to post-event analysis, due to the dependence on large faults for triggering noticeable transient responses. 
By considering the \textit{small-signal analysis} in power system dynamics \cite[Ch.~12]{kundur1994power}, ambient synchrophasor data have been popularly used for estimating individual oscillation modes; see for example \cite{pierre1997initial,zhou2009electromechanical,ning2014two} and references therein. To recover the full system's dynamic responses\footnote{{In this paper, the frequency, angle, or line flow response refers to the respective impulse response of frequency, angle, or line flow in the time domain. This terminology is consistent with the generator frequency response in power system dynamics \cite{nerc_freq_resp}, which is different from frequency-domain  analysis of signals and systems.}}, the statistical information of ambient data has been utilized to estimate the dynamic state Jacobian matrix \cite{wang2017pmu,sheng2020online}. However, these approaches to estimate the Jacobian matrix require the availability of state measurements at a majority of grid locations and cannot cope with limited PMU deployments thus far. 

Our goal is to develop a general recovery framework for inertia-based dynamic responses in the small-signal regime by using ambient synchrophasor data. We propose a cross-correlation-based approach to process ambient synchrophasor data and utilize a synthetic Texas Interconnection (TI) system to demonstrate its effectiveness and advantages. The proposed approach is quite general and flexible in data types or PMU locations, as it can incorporate any frequency, angle, or line flow data streams from any pair of bus locations. Targeting at the small-signal regime, we first study the second-order swing dynamics to establish the theoretical underpinnings of the proposed data-driven approach. While this analytical result requires homogeneous damping among key inter-area modes, such condition is reasonable for wide-area interconnection \cite{cui2017inter} and corroborated by numerical case studies as well. An ambient data-driven algorithm is accordingly developed to recover the dynamic responses to a disturbance from any input location. Going beyond the theoretical analysis, the proposed approach is numerically validated on a synthetic TI system that includes realistic nonlinear higher-order dynamics. The case study results strongly support the effectiveness of our approach in recovering dynamic responses and its advantage of incorporating general types of synchrophasor data.

The present work extends a prior work \cite{huynh2018data} on ambient frequency data analytics to comprehensively encompass ambient data of angles and line flows. This extension is valuable for real-world applications, as PMUs are typically installed at buses or branches instead of measuring generators directly. Additionally, angle/power data are known to have higher accuracy than frequency data, as the high-pass filtering for PMUs to generate frequency data  \cite{pmu_report2020} can adversely affect data quality as corroborated by our synthetic TI system studies later on. Another related work \cite{jalali2022inferring,GP4dynamicsCDC21} has recently proposed a Gaussian process (GP)-based approach for inferring data streams at multiple locations assuming the system model is known.

The main contribution of this work is two-fold. First, we establish the equivalence between model-based dynamic responses and the cross-correlation of various ambient data. This equivalence builds upon the second-order swing dynamics and considers reasonable assumptions for wide-area interconnections. Second, we develop a fully data-driven algorithm to recover system responses that can incorporate all types of data at minimal PMU deployment. The proposed recovery algorithm requires no knowledge of the actual system model or parameters and is applicable to synchrophasor data streams from any pair of grid locations. 

The rest of paper is organized as follows.  Section \ref{sec:ps} introduces the problem and the synthetic TI system. Section \ref{sec:theoretical_results} establishes the equivalence between small-signal dynamic responses and the cross-correlation of
ambient data, which enables Section \ref{sec:algorithm} to develop the proposed data-driven algorithm. Section \ref{sec:numerical_results} demonstrates its validity and advantages on the 2000-bus synthetic TI systems with realistic dynamic modeling and PMU data generation. The work is concluded in Section~\ref{sec:con}.


\section{Problem Statement}
\label{sec:ps}



{We aim to develop a data-driven framework to infer the dynamic responses of the power system using ambient synchrophasor data.} First, we provide a brief review of small-signal dynamic modeling. The approximate swing equation \cite[Ch.~9]{arthur2000power} for a power system with $N$ generators is given by
\begin{align}
    \Dot{\bbdelta} &= \bbomega \\
    \bbM\Dot{\bbomega} &= -\bbK\bbdelta - \bbD\bbomega + \bbu 
\label{eq:swing1}
\end{align}
where states $\bbdelta,\bbomega \in \mathbb R^N$ are the rotor angle and speed (frequency) vectors, respectively. The diagonal matrices $\bbM$ and $\bbD$ contain respectively the generator inertia and damping constants, while $\bbK$ stands for the power flow Jacobian matrix evaluated at the given operating point. The swing dynamics in \eqref{eq:swing1} is a linearized approximation of the actual power system dynamics, which are nonlinear and contain higher-order components (e.g., governors and exciters) to be discussed later. 

Model~\eqref{eq:swing1} can also be written in an equivalent second-order form
\begin{align}
    \bbM\Ddot{\bbdelta} + \bbD\Dot{\bbdelta} + \bbK\bbdelta = \bbu.
    \label{eq:swing2}
\end{align}
Note that all state variables in the linear time-invariant (LTI) system \eqref{eq:swing2} are represented by the \textit{deviations} from their steady-state values. For simplicity, the term deviations will be dropped henceforth. 

Under this LTI approximation, the system's dynamic response to a power deviation input $\mathbf{u}$ is fully characterized through the \textit{impulse response}. Let $T_{u_k,\delta_\ell}(\tau)$ denote the impulse response of the target rotor angle $\delta_\ell$ from input source at $u_k$. Similar notations will used for other target variables; for example, the impulse response of frequency $\omega_\ell$ from input source $u_k$ is denoted by $T_{u_k,\omega_\ell}(\tau)$. Leveraging the unique properties of the swing dynamics, we show that the aforesaid impulse responses between grid locations $(k,\ell)$ can be recovered by simply cross-correlating synchrophasor data collected only at $(k,
ell)$ under ambient conditions. Such methodology does not require to know the system model or probe the system with any particular inputs. The term \emph{ambient conditions} here refers to random perturbations of active power injections (due to load or generation variations) giving rise to a ``white-noise'' type of input $\bbu(t) = \bbnu(t)$ satisfying~\cite{ning2014two,wang2017pmu} 
\begin{align}
    & \mathbb{E}\left[\bbnu(t)\right] = \mathbf 0 \notag\\
    & \mathbb{E}\left[\bbnu(t)  \bbnu^\top(t-\tau) \right] = \bbSigma \Delta(\tau)
    \label{eq:input_var}
\end{align}
where $\Delta(t)$ is the Dirac delta function. Under the input in \eqref{eq:input_var}, the corresponding ambient state/output will be denoted by hatted symbols, such as $\hat{\delta}_\ell(t)$ and $\hat{\omega}_\ell(t)$. The cross-correlation of ambient angle signals is defined as
\begin{align}
    C_{\hat{\delta}_k\hat{\delta}_\ell}(\tau) &\triangleq \lim_{T\rightarrow\infty} \frac{1}{2T} \int_{-T}^{T} \hat{\delta}_k(t)\hat{\delta}_\ell(t-\tau)dt \notag \\
     &= \mathbb{E}\left[\hat{\delta}_k(t)\hat{\delta}_\ell(t-\tau)\right]
    \label{eq:cross_correlation}
\end{align}
where the second equality is due to the stationary input process [cf.~\eqref{eq:input_var}], as detailed in \cite[Ch.~9]{stirzaker1992probability}.
The same equivalence holds for other types of cross-correlations to be discussed later on. Note that PMUs are installed at terminal buses, and thus, do not directly measure the internal generator states $\hat{\delta}_\ell$ or $\hat{\omega}_\ell$. Nevertheless, the terminal angle $\hat{\theta}_n(t)$ and frequency $\frac{d\hat{\theta}_n(t)}{dt}$ measured at generator bus $n$ are excellent surrogates of their internal counterparts~\cite{markham2014electromechanical}.

\begin{figure}[t]
  \centering
  \includegraphics[width=.6\linewidth]{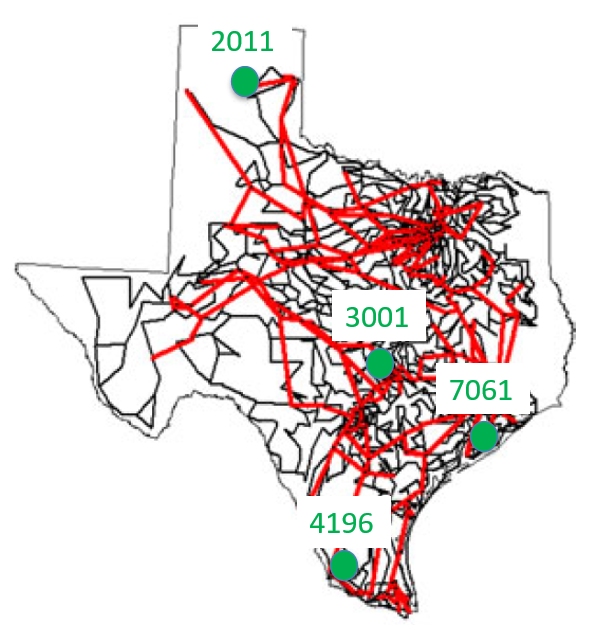}
  \caption{The topology of the 2000-bus synthetic Texas Interconnection (TI) system \cite{birchfield2016grid,idehen2020large}. There are a total of 99 PMUs placed in the system, all at 345kV buses indicated by the red-colored high-voltage edges, while gray-colored edges represent the 115kV lines. Synthetic PMU data at the four marked buses will be used later.}
  \label{fig:texas2000}
\end{figure}

This work will utilize a synthetic TI system to showcase the effectiveness of the proposed data-driven method. As a large-scale test case, the synthetic TI system is of high interest for numerical studies with inter-area modes in the range of $[0.62,~0.73]$~Hz \cite{nerc_report2021}. To overcome the issue of  actual synchrophasor data access, we use the dynamic data generated on a 2000-bus synthetic TI system \cite{birchfield2016grid}, as illustrated  in Fig.~\ref{fig:texas2000}. As a realistic representation, this synthetic TI system has 1500 substations, 500 generators, approximately 50 GW of peak load, and 3206 transmission lines. The dynamic components  including generators and loads have been modeled based on the actual Texas system, while the dynamic responses have been validated with the actual data~\cite{idehen2020large}.

\section{Ambient Data Analytics}
\label{sec:theoretical_results}

Our goal is to establish the theoretical equivalence between the actual system responses and the cross-correlation of phasor data collected under ambient grid conditions based on analyzing the swing equation \eqref{eq:swing2}. To this end, the following simplifying assumptions are posed to decouple the system into independent modes: 

\begin{assumption}
The generator inertia and damping constants are homogeneous; namely $\bbD = \gamma \bbM$ for a constant $\gamma >0$.
\label{assump1}
\end{assumption}

\begin{assumption}
The power flow Jacobian $\bbK$ is a symmetric Laplacian, and hence, positive semidefinite (PSD) matrix.
\label{assump2}
\end{assumption}

Assumption~(AS\ref{assump1}) holds if damping and inertia parameters scale proportionally for each generator. This assumption has been frequently adopted for approximating power system dynamics~\cite{Paganini19,huynh2018data,jalali2022inferring}. As for (AS\ref{assump2}), if transmission lines are all purely inductive (lossless) and loads are of constant power outputs in \eqref{eq:swing1}, matrix $\bbK$ becomes symmetric and (AS\ref{assump2}) would hold.  Certain load models (such as constant-current ones \cite{osti_1004165}) could affect the dynamic models, but may have minimal impact on the symmetry of $\bbK$. This is because under small-signal analysis, the system voltages tend to be steady, thus leading to minimal load power changes. Assumptions (AS\ref{assump1})--(AS\ref{assump2}) are adopted only to establish the analytical results, but will be waived during our numerical tests; see Remark~\ref{rmk:model} for further discussions on generalizability.

Under (AS\ref{assump1})-(AS\ref{assump2}), the multi-input multi-output (MIMO) LTI system in \eqref{eq:swing2} decouples into $N$ single-input single-output (SISO) systems. This is achieved upon linearly transforming system states as $\bbdelta = \bbV \bbz$, where matrix $\bbV$ is specified by the generalized eigenvalue problem $\bbK\bbV = \bbM\bbV\bbLambda$ and diagonal matrix $\bbLambda$ carries the $N$ eigenvalues $\lambda_{i} \geq 0$. Matrix $\bbV$ is $\bbM$-orthonormal and satisfies \cite[Sec.~5.2]{strang2006linear}:
\begin{align}
\bbV^\top \bbM\bbV = \bbI~~~\textrm{and}~~~\bbV^\top \bbK\bbV = \bbLambda. \label{eq:MK}
\end{align}
Substituting \eqref{eq:MK} into \eqref{eq:swing2} and utilizing (AS\ref{assump1}) lead to a completely decoupled second-order system, given by 
\begin{align}
    \Ddot{\bbz} + \gamma\Dot{\bbz} + \bbLambda \bbz = \bbV^\top \bbu. 
    \label{eq:swing_modes}
\end{align}
Solving for each independent mode $z_i$ in \eqref{eq:swing_modes} gives rise to the impulse responses \cite{huynh2018data}
\begin{align}
    T_{u_k,\omega_\ell}(\tau) &= \sum_{i=1}^N V_{ki}V_{\ell i} ~\eta_i \left(c_i e^{c_i\tau} -d_i e^{d_i\tau} \right) \label{eq:freq_impz}\\
    T_{u_k,\delta_\ell}(\tau) &= \sum_{i=1}^N V_{ki}V_{\ell i}~\eta_i \left( e^{c_i\tau} - e^{d_i\tau} \right)
    \label{eq:ang_impz}
\end{align}
%
where $V_{ki}$ is the $(k,i)$-th entry of matrix $\mathbf{V}$, and with the mode-associated complex-valued parameters
\begin{align*}
&c_i= \frac{-\gamma + \sqrt{\gamma^2 - 4\lambda_i}}{2},~~ d_i= \frac{-\gamma - \sqrt{\gamma^2 - 4\lambda_i}}{2}\\
&\text{and}~~\eta_i =\frac{1}{\sqrt{\gamma^2 - 4\lambda_i}}
\end{align*}
for $i=1,\ldots,N$. Upon obtaining the impulse responses in \eqref{eq:freq_impz}--\eqref{eq:ang_impz}, we will exploit the structure therein for ambient data processing. Ambient conditions are formally defined as:
\begin{assumption}
Ambient data during nominal operations are generated by random noise $\bbnu(t)$ satisfying \eqref{eq:input_var} with variance proportional to inertia as $\bbSigma = \alpha \bbM$ with a constant $\alpha>0$.
\label{assump3}
\end{assumption}

Assumption (AS\ref{assump3}) is introduced to guarantee that all modes in \eqref{eq:swing_modes} are equally and independently excited, thanks to the diagonalization $\bbV^\top\bbSigma \bbV = \alpha \bbI$ [cf.~\eqref{eq:MK}]. Note that real-world power systems may not perfectly balance generation inertia with load variability, as most types of generation are placed based upon resource availability. However, for a large interconnection (AS\ref{assump3}) could hold broadly over all control areas, instead of at every location. Furthermore, to deal with lowering inertia in current power systems, the placement of virtual synchronous generators~\cite{arghir2018grid} and virtual inertia \cite{poolla2019placement} tend to account for load variability. To sum up, even though (AS\ref{assump3}) may not hold perfectly for actual grids, it aims to ensure homogeneous damping among significant inter-area modes, which is reasonable for a wide-area system as shown by actual synchrophasor data analysis in \cite{cui2017inter}.

Reference~\cite{huynh2018data} established that under (AS\ref{assump1})--(AS\ref{assump3}) the frequency response between any pair of buses $(k,\ell)$ can be recovered up to a scalar uncertainty in a model-agnostic fashion by simply cross-correlating frequency data collected during ambient conditions. 

\begin{lemma}{\emph{(Frequency Response~\cite{huynh2018data})}}\label{prop:freq}
Under (AS\ref{assump1})-(AS\ref{assump3}), the cross-correlation of ambient frequency $\hat{\omega}_k$ and $\hat{\omega}_\ell$ is related to the frequency response as
\begin{align}
    T_{u_k,\omega_\ell}(\tau) = -\frac{2\gamma}{\alpha} C_{\hat{\omega}_k,\hat{\omega}_\ell}(\tau). \label{eq:freq} 
\end{align}
\end{lemma}

This is an interesting result as it does not hold for a general LTI system. Surprisingly, it does hold for swing dynamics because \eqref{eq:swing2} decouples into $N$ scalar systems under (AS\ref{assump1})--(AS\ref{assump2}), while (AS\ref{assump3}) guarantees the theoretical equivalence.
Thanks to this unique property, one can recover pairwise frequency responses without knowing the system model by simply cross-correlating frequency datastreams. Nonetheless, in practice synchronized frequency data are not as accurate as angle and line flow measurements because of the PMU internal low-pass filtering~\cite{pmu_report2020}. To this end, we generalize this previous result by incorporating general types of PMU measurements, such as bus angle and line flow readings. Heed that the extension is technically challenging, as states at buses are not explicitly reflected in \eqref{eq:swing2}. Furthermore, it is also of high practical value as PMUs are typically installed at buses or branches rather than measuring generator states directly.

To generalize Lemma \ref{prop:freq}, we propose new approaches for directly processing ambient angle and power measurements provided by PMUs, as follows. 

\begin{proposition}{\emph{(Angle Response)}}\label{prop:ang}
Under (AS\ref{assump1})-(AS\ref{assump3}), the cross-correlation of ambient angles $\hat{\delta}_k$ and $\hat{\delta}_\ell$ is related to the angle response as
\begin{align}
T_{u_k,\delta_\ell}(\tau) &= -\frac{2\gamma}{\alpha} \frac{d}{d\tau}C_{\hat{\delta}_k,\hat{\delta}_\ell}(\tau)= -\frac{2\gamma}{\alpha} C_{\hat \omega_k,\hat\delta_\ell}(\tau). 
\label{eq:ang_resp} 
\end{align}
\end{proposition}

\begin{proof}
The ambient angle is the convolution of input noise $\bbnu(t)$ and the impulse response  in \eqref{eq:ang_impz}. Hence, we can define the vector $\bbh_k(t) = \left[ V_{ki}\eta_i \left( e^{c_it} - e^{d_it} \right) \right]_{N\times1}$ and show that 
\begin{align*}
    &C_{\hat{\delta}_k,\hat{\delta}_\ell}(\tau) 
    = \int_0^\infty dt_1 \int_\tau^\infty dt_2 \nonumber\\
    & \cdot\bbh_k(t_1)^\top \bbV^\top \mathbb{E}\left[\bbnu(t-t_1)\bbnu(t-\tau-t_2)^\top\right] \bbV \bbh_\ell(t_2) \nonumber\\
    =& -\alpha\sum_{i=1}^N V_{ki}V_{\ell i}\eta_i^2 \\
    & \; \cdot\left[ \left(\frac{1}{2c_i}+\frac{1}{\gamma}\right)e^{c_i  \tau} 
    + \left(\frac{1}{2d_i}+\frac{1}{\gamma}\right)e^{d_i  \tau} \right]
\end{align*}
where the second equality uses the white-noise property  and the diagonalization $\bbV^\top\bbSigma \bbV = \alpha \bbI$ in (AS\ref{assump3}). Taking its derivative and utilizing the relations among $c_i$, $d_i$ and $\eta_i$ lead to the equivalence between $\frac{d}{d\tau}C_{\hat{\delta}_k,\hat{\delta}_\ell}(\tau)$ and $T_{u_k,\delta_\ell}(\tau)$ as in \eqref{eq:ang_impz}. To obtain the result for $C_{\hat{\omega}_k,\hat{\delta}_\ell}(\tau)$, one can use the fact that $\hat{\omega}_k(t) = \frac{d}{dt} \hat{\delta}_k(t)$ to show the relation between the two cross-correlations. \qed
\end{proof}

Proposition \ref{prop:ang} leads to a corollary on recovering the frequency response due to its relation to angle response as in \eqref{eq:freq_impz}-\eqref{eq:ang_impz}.
\begin{corollary}{\emph{{(Frequency Response)}}}\label{col:freq_ang}
Under (AS\ref{assump1})-(AS\ref{assump3}), the cross-correlation of ambient angle $\hat{\delta}_k$ and $\hat{\delta}_\ell$ is related to the frequency response as
\begin{align}
 T_{u_k,\omega_\ell}(\tau) = -\frac{2\gamma}{\alpha}\frac{d^2}{d\tau^2}C_{\hat{\delta}_k,\hat{\delta}_\ell}(\tau) &= -\frac{2\gamma}{\alpha}\frac{d}{d\tau}C_{\hat{\omega}_k,\hat{\delta}_\ell}(\tau). 
\label{eq:freq_resp2}
\end{align}
\end{corollary}

Proposition \ref{prop:ang} and Corollary \ref{col:freq_ang} neatly extend the ambient frequency data analysis to that for angle data. The key difference is the differentiation needed for achieving the original model coefficients in \eqref{eq:freq_impz}--\eqref{eq:ang_impz}. As in~\eqref{eq:freq}, there exists a scaling difference between the cross-correlation and the dynamic response, which could be estimated based on past event analysis \cite{peydayesh2017simplified}. 

As mentioned in Section \ref{sec:ps}, PMUs actually measure bus- or branch-related quantities, instead of generator internal states. Hence, we will generalize the  cross-correlation equivalence to include bus angle $\theta_n$ and line power flow $p_{nm}$ from bus $n$ to $m$. Recall that in small-signal analysis, the linearized power flow equation in \eqref{eq:swing1} admits that the output bus angle is linearly related to the state as
\begin{align}
    \theta_n(t) =\bba_{n}^\top \bbdelta(t) = \textstyle \sum_{\ell=1}^N a_{n\ell} \delta_\ell(t),
    \label{eq:lin_bus_ang}
\end{align}
and similarly for line flow $p_{nm}(t)$. This linearity is instrumental for extending the analysis to cross-correlating ambient angle $\hat{\theta}_n(t)$, as  the corresponding dynamic response and cross-correlation also follow this linear transformation; that is, 
\begin{align}
    T_{u_k,\theta_{n}}(\tau) &= \textstyle \sum_{\ell=1}^N a_{n\ell} T_{u_k,\delta_{\ell}}(\tau),\\ 
    {C}_{\hat{\omega}_k,\hat{\theta}_n}(\tau) &= \textstyle\sum_{\ell=1}^N a_{n\ell} C_{\hat{\omega}_k,\hat{\delta}_{\ell}}(\tau). 
\end{align}
Observe that the generator bus frequency or equivalently the derivative of generator bus angle, can accurately approximate the connected rotor speed, as discussed in Section~\ref{sec:ps}. For brevity, we use $\hat{\theta}_k(t)$ as the observed angle at the bus closest to input $u_k$ and establish the following result using the observed ambient angle (frequency).  
\begin{proposition}{\emph{(Bus Angle Response)}} \label{prop:bus_angle}
Under (AS\ref{assump1})-(AS\ref{assump3}), the cross-correlation of ambient bus angle data $\hat{\theta}_k$ and $\hat{\theta}_n$ is related to the bus angle response as
\begin{align}
T_{u_k,\theta_{n}}(\tau) = -\frac{2\gamma}{\alpha}\frac{d}{d\tau} {C}_{\hat{\theta}_k,\hat{\theta}_n}(\tau) = -\frac{2\gamma}{\alpha}{C}_{\hat{\omega}_k,\hat{\theta}_n}(\tau). \label{eq:bus_angle} 
\end{align}
\end{proposition}
Similar to bus angle, the ambient line flow measurements can be used to recover its response as well. 
\begin{proposition}{\emph{(Line Flow Response)}}\label{prop:flow}
Under (AS\ref{assump1})-(AS\ref{assump3}), the cross-correlation of ambient line flow $\hat{p}_{nm}$ and angle $\hat{\theta}_k$ is related to the line flow response as
\begin{align}
T_{u_k,p_{nm}}(\tau) = -\frac{2\gamma}{\alpha}\frac{d}{d\tau} {C}_{\hat{\theta}_k,\hat{p}_{nm}}(\tau) = -\frac{2\gamma}{\alpha}{C}_{\hat{\omega}_k,\hat{p}_{nm}}(\tau). \label{eq:freq_flow}
\end{align}
\end{proposition}


\begin{remark}{\emph{(Generalizability)}}\label{rmk:model} 
    Although assumptions (AS\ref{assump1})--(AS\ref{assump3}) are necessary for the analytical equivalence, they can be well relaxed to match real-world grid conditions. A key premise for our analysis is that under (AS\ref{assump3})  the  modes are equally and independently excited, such that the cross-correlation output would maintain the same coefficients for all modes. In practice, inter-area modes are more evident than local intra-area modes in a wide-area interconnection \cite[Ch.~10]{chow2013power}. As long as the dominant inter-area modes are equally excited, the equivalence results should hold as well. Our numerical studies have demonstrated the cross-correlation outputs can approximately recover the dynamic responses even when (AS\ref{assump1})-(AS\ref{assump3}) are violated, including using higher-order generator dynamics and perturbing load demands instead of generator inputs for ambient conditions.
\end{remark}

\section{The Recovery Algorithm}
\label{sec:algorithm}

\begin{figure}[tb!]
  \centering
  \includegraphics[width=\linewidth]{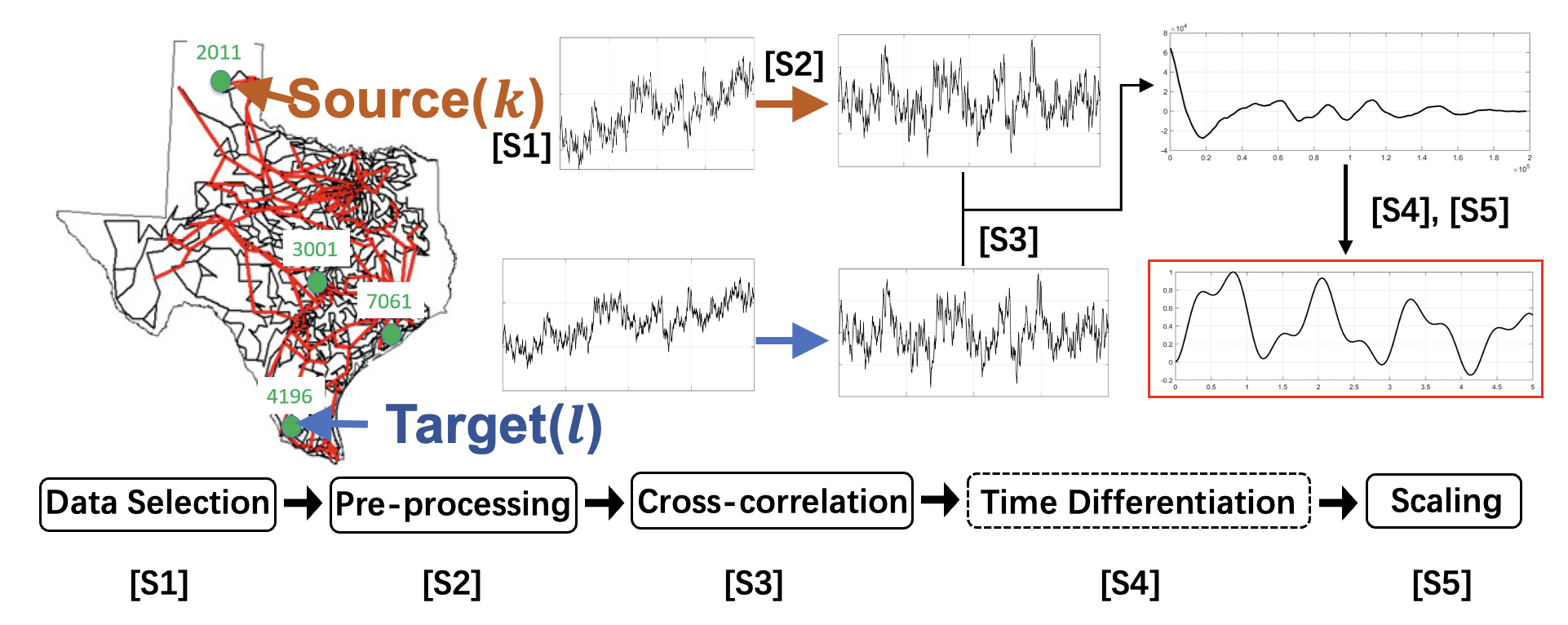} 
  \caption{The proposed 5-step algorithm to recover the dynamic responses using ambient synchrophasor data at any two  locations (source and target).}
  \label{fig:alg}
  \vspace*{-2mm}
\end{figure}

\begin{figure*}[!t]
    \centering
    \includegraphics[width=160mm]{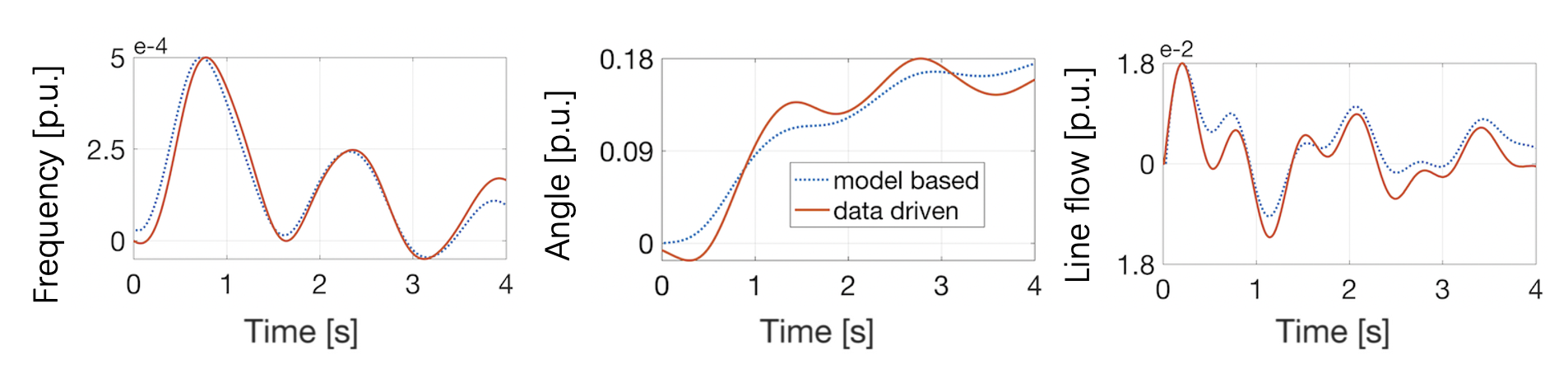}
    \caption{Comparison of model-based and data-driven dynamic responses for the WSCC 9-bus test case under the sixth-order generator model, non-uniform damping, and random load perturbations. The disturbance input is set to be $u_2$ at generator bus 2, while the responses are compared for frequency \emph{(left)} rotor angle \emph{(middle)} at generator 1, as well as line flow \emph{(right)} between bus 7 and 8.}
    \label{fig:6th_nonunif_dyn_rsp_load}
\end{figure*}

Using the equivalence results, we next describe the algorithm for recovering dynamic responses. The implementation is flexible in the types or locations of data. Typically, PMUs are installed at critical substations with large generation or power flow within each control area. As shown by the cross-correlation equivalence results, the PMU data streams from any two locations can be used for recovering the dynamic responses, including measurements from PMUs not located at generators. Motivated by the idea, we propose the following five-step dynamics recovery algorithm illustrated in Fig.~\ref{fig:alg}.

\begin{enumerate}
    \item[\textbf{[S1]}] \textbf{(Data selection)}  For recovering response $T_{k,\ell}(\tau)$, select the raw data at any source (${x}_k$) and any target (${x}_\ell$) locations, from the closest PMUs in terms of electrical distance~\cite{cotilla2013multi}. For input generator $k$, this could be frequency/angle data from the substation directly connected to it. For the target location, it can be system-level outputs, such as bus frequency, bus angle, or line flows. 
      
    \item[\textbf{[S2]}] \textbf{(Preprocessing)} Preprocess the raw data to obtain the proper ambient response signals.  
    We use a bandpass filter to find the detrended signals $\hat{x}_k$ and $\hat{x}_\ell$. 
    As inter-area oscillation modes are of high interest, the recommended passband of the filter is between {$[0.1,~0.8]$Hz \cite{nerc_report2021}}. 
    
    \item[\textbf{[S3]}] \textbf{(Cross-correlation)}  With a sampling period $T_s$ and thus a total of $\mathcal{M} = \lfloor T/T_s\rceil$ samples, compute the discrete-time version of the cross-correlation as
          \begin{align*}
          C_{k, \ell}[\tau] = \frac{1}{\mathcal{M}}
          \sum_{m=1}^{\mathcal{M}}\hat{x}_k[m]\hat{x}_\ell[m-\tau].
          \end{align*}
   
  \item[\textbf{[S4]}] \textbf{(Differentiation)} Take the numerical difference of $C_{k, \ell}[\tau]$ depending on the type of dynamic responses of interest, e.g., recovering the frequency response from ambient angle data (cf.~Corollary \ref{col:freq_ang}).   
  
    \item[\textbf{[S5]}] \textbf{(Scaling)}  If the frequency nadir point is known, one can scale the cross-correlation output to match it. 
    Otherwise, the recovered responses will be used for evaluating the propagation time. 
\end{enumerate}

The proposed algorithm can directly be applied to infer the dynamic response from any source to another target from the ambient measurements at both locations. It is also computationally efficient. For $\mathcal{M}$ samples, the computation is mainly due to [S3] at $\ccalO(\mathcal{M}^2)$. 
\begin{table}[t]
\centering
\caption{Normalized MSE of recovering different dynamic responses for the WSCC 9-bus test case.} 
\begin{tabular}{l|ll}
\hline
& 2nd unif. & 6th non-unif. \\ \hline
Frequency & 0.25    & 0.26   \\
Angle      & 0.12 & 0.10 \\
Flow       & 0.20    & 0.37  \\  
\hline 
\end{tabular}
\label{table:wscc9-accuracy}
\end{table}

\subsection{{Validation on the WSCC 9-Bus System}}
 \label{subsec:2nd_unif}
To demonstrate the effectiveness of the proposed algorithm, we first use the small WSCC 9-Bus System.\footnote{All codes, and results are available at: \newline \noindent \url{https://github.com/ShaohuiLiu/dy_resp_pkg_new}.} This system has 3 generators, 3 loads, and 9 transmission lines. Impulse responses as well as ambient data have been generated in MATLAB using PSAT~\cite{milano2005open}. To emulate realistic grid operations, we set the disturbance location at one generator bus and use load perturbations for ambient data generation. For the impulse response, we have run time-domain simulation based on the nonlinear models with a very short ``impulse''-like input $u_2$ at the generator bus 2. To generate ambient signals, all loads have been perturbed by random white noise generated by MATLAB's function \texttt{randn} to mimic (AS\ref{assump3}). The sampling rate is set to 100Hz using the simulation time-step of $dt=0.01s$. Estimation accuracy is quantified by the normalized mean squared error (MSE)
\begin{align}
 \frac{\| T_{u_k,x_n} - {C}_{k,n}\|_2}{\| T_{u_k,x_n}\|_2} 
\label{eq:est_err}
\end{align}
where $T$ or ${C}$ here stand for the model-based or data-driven response, respectively, normalized by its maximum absolute value.

We first validated the proposed algorithm under the classical second-order generator model and uniform damping condition ($\gamma = 0.2$) following (AS\ref{assump1}). 
Note that the power flow Jacobian $\bbK$ is not perfectly symmetric as needed in (AS\ref{assump2}), as transmission lines are not purely inductive. 
Despite this slight violation of (AS\ref{assump2}), we have observed that the recovered dynamic responses match well with the model-based ones, as shown in Table~\ref{table:wscc9-accuracy} by averaging over all system locations. The small estimation error is attributed primarily to linearization and the asymmetry of matrix $\bbK$.

We further considered the original setting with the sixth-order generator model that includes controllers like governor, exciter, and power system stabilizer. We have also changed the damping to be non-uniform ($\gamma\in [0.1,0.3]$). All these settings reflect realistic power system dynamics with nonlinearity and ambient conditions. Figure~\ref{fig:6th_nonunif_dyn_rsp_load} compares the responses for frequency, angle, and line flow outputs at selected locations. The MSE values between model-based and data-driven dynamic responses as given in Table~\ref{table:wscc9-accuracy} confirm the effectiveness of our novel algorithm, despite that this test has significantly departed from assumptions (AS\ref{assump1})-(AS\ref{assump2}). 

We also compared the recovery performance by comparing the corresponding oscillation frequency and damping coefficient, both estimated by using the logarithmic decrement method \cite[Ch.~2]{meirovitch2001fundamentals}. Table~\ref{table:wscc9-damping} lists these estimated parameters for the data- and model-based approaches alike, confirming  the effectiveness of our proposed solution in terms of recovering dynamic parameters. Therefore, relaxing our analytical assumptions to more realistic grid conditions has led to minimal effect on the recovery performance. To sum up, the proposed framework can well recover different types of system responses based on simulated tests on the small 9-bus system.

\section{Synthetic TI Case Study}
\label{sec:numerical_results}


\begin{table}[t]
\centering
\caption{Comparison of dynamic parameters of oscillation frequency and damping coefficient estimated by the model-based and data-driven approaches for the WSCC 9-bus test case under the
sixth-order generator model, non-uniform damping, and random load perturbations.}
\begin{tabular}{l|ll|ll}
\hline
& \multicolumn{2}{c}{Frequency}   \vline&  \multicolumn{2}{c}{Damping}  \\
& data & model & data & model \\ \hline
$\omega_1$ & 0.642 & 0.634 & 0.290 & 0.315 \\
$\omega_2$ & 0.645 & 0.638 & 0.309 & 0.353 \\
$\omega_3$ & 0.653 & 0.655 & 0.436 & 0.406  \\  
\hline 
\end{tabular}
\label{table:wscc9-damping}
\end{table}


This section presents the  case study results for the synthetic TI system presented in Section~\ref{sec:ps}. The case study uses synthetically generated ambient synchrophasor data, which can demonstrate the importance of using angle/line flow data over frequency data in real-world grid conditions. Ambient dynamics were simulated as follows: \emph{i)} Periodic variations at 5s- and 7s-intervals have been set up for loads and generators, respectively; \emph{ii)} A simulation time step of one-quarter cycles, with power flow result storage every 8 time steps, is used by the solver specification, close to PMU data rates of 30 samples per second. For ambient data, we consider a total duration of $10$min data with a sampling rate of $30$Hz. The simulated ambient data have been further processed to produce synthetic synchrophasor data, by adding {0.002\% random measurement noise to all data streams}. In addition, frequency data are further filtered based on the actual PMU processing method as described in \cite{idehen2020large}, to mimic the statistics of actual frequency data. This filtering step for synthetic frequency data has made it less reliable for recovering dynamic responses, as detailed soon.

As illustrated in Fig.~\ref{fig:texas2000}, we pick 4 buses in different regions of TI. The distance between the north-region Bus 2011 and south-region Bus 4196 is around 670 miles, while the distance from north- to coastal region Bus 7061 is about 450 miles. The frequency responses are very similar within the TI system due to system size and frequency control designs~\cite[Ch.~10]{ercot_control2016}. To compare responses, we selected Bus 2011 in the north region as the input/source location, and the three other buses (3001, 4196, and 7061) as the output/target locations. 

We first compared the recovered frequency responses obtained by both ambient frequency data and angle data, as plotted in Fig.~\ref{fig:freq_resp}. To process ambient data, we have set the filter passbands in  \textbf{[S2]} to be $[0.3,~0.75]$Hz to include inter-area modes. As for the ambient angle data, we also need to compute a reference angle by taking the average over all recorded angle data, as discussed in \textbf{[S1]}. This reference angle is subtracted from the ambient angle data prior to bandpass filtering. 
The simulated frequency data have been used as the benchmark for evaluating synthetic angle and frequency data. Using the simulated frequency data, the proposed cross-correlation outputs show very similar frequency responses at all locations, except for some minor time lags among the first nadir points, as shown by Fig.~\ref{fig:freq_resp}(a). The time of frequency nadir points as estimated by our proposed algorithm is listed in Table~\ref{table:freq_lag}. A closer look at the time lags confirms with the nominal speed for electromechanical wave propagation, which is around 200-1,000~mi/sec for typical systems \cite{alharbi2020simulation}. 
The synthetic angle data produce very similar frequency responses in Fig.~\ref{fig:freq_resp}(b), corroborating the effectiveness of the proposed framework. However, due to PMUs' signal processing step in filtering frequency data, the synthetic frequency data have led to highly inaccurate frequency responses which clearly lack in synchronization, as shown in Fig.~\ref{fig:freq_resp}(c). This comparison speaks for the importance of the proposed extension over the frequency data only approach of \cite{huynh2018data}. 

\begin{table}[!t]
\centering
\caption{The time of nadir points and their lags at the four locations in the 2000-bus synthetic TI system along with the propagation speed, as estimated by the proposed cross-correlation approach using simulated frequency data.}
\begin{tabular}{l|llll}
\hline 
Bus Index   & 2011 & 3001 & 7061 & 4196 \\ \hline
Distance/mi & 0    & 370  & 535  & 670  \\
Time/s      & 0.93 & 1.17 & 1.37 & 1.43 \\
Lag/s       & 0    & 0.24 & 0.44 & 0.50 \\  [1ex] 
\hline 
\end{tabular}
\label{table:freq_lag}
\vspace*{-2mm}
\end{table}



We further evaluated the angle response recovery from synthetic angle data, as plotted in Fig.~\ref{fig:ang_resp}. Similar to frequency responses, the underlying oscillation modes in all angle responses are very similar among all locations. Nonetheless, the time shifts among them are visible, consistent with the fact that their relative distances to the source of input are different. Thus, the proposed data-driven algorithm has achieved accurate and consistent recovery of grid dynamics from ambient PMU data that have been realistically generated. Thanks to the higher accuracy of ambient angle data over frequency data, our proposed work exhibits enhanced practical value than the previous frequency-based method with improved performance and reliability.

\begin{figure}[t!]
        \vspace*{-2mm}
        \centering
        \hspace*{-8mm}
        \includegraphics[width=100mm]{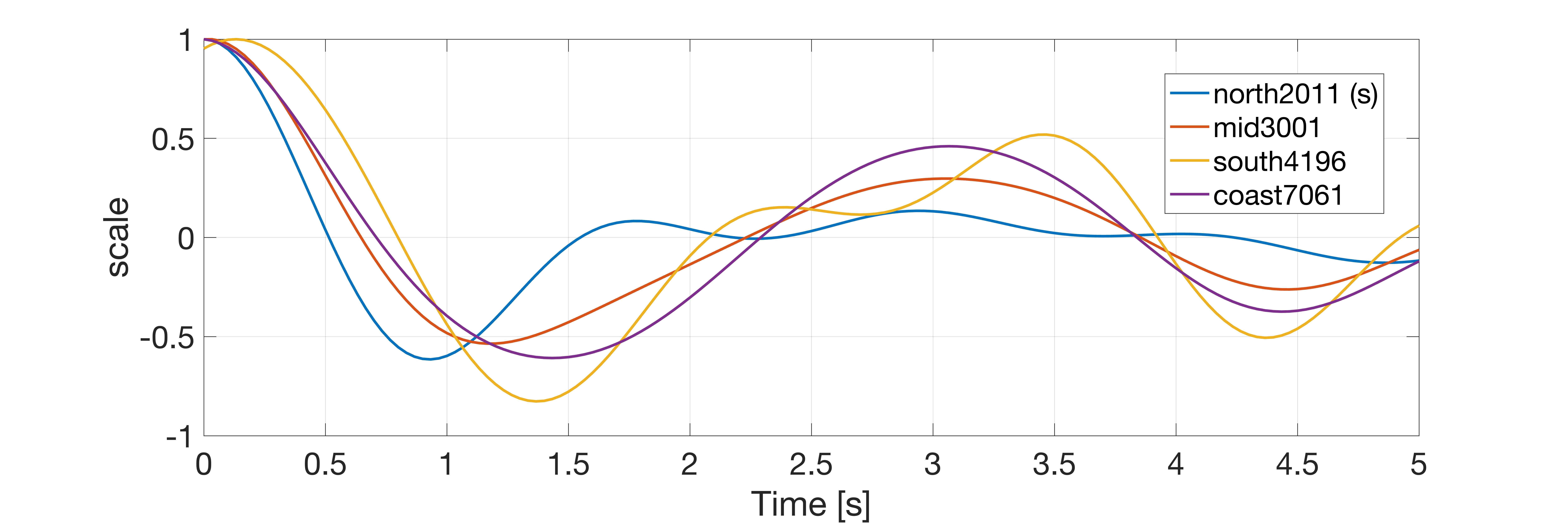} 
        \centerline{(a) Simulated frequency data}
        \hspace*{-8mm}
        \includegraphics[width=100mm]{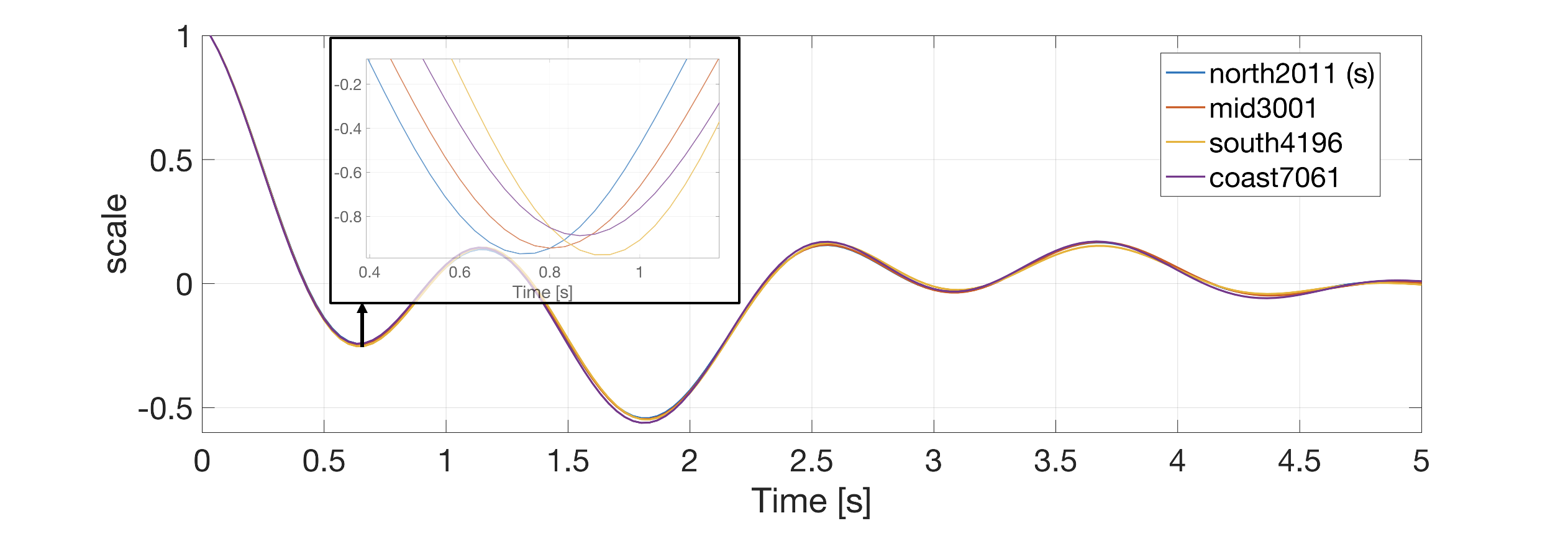}
        \centerline{(b) Synthetic angle data}
        \hspace*{-8mm}
        \includegraphics[width=100mm]{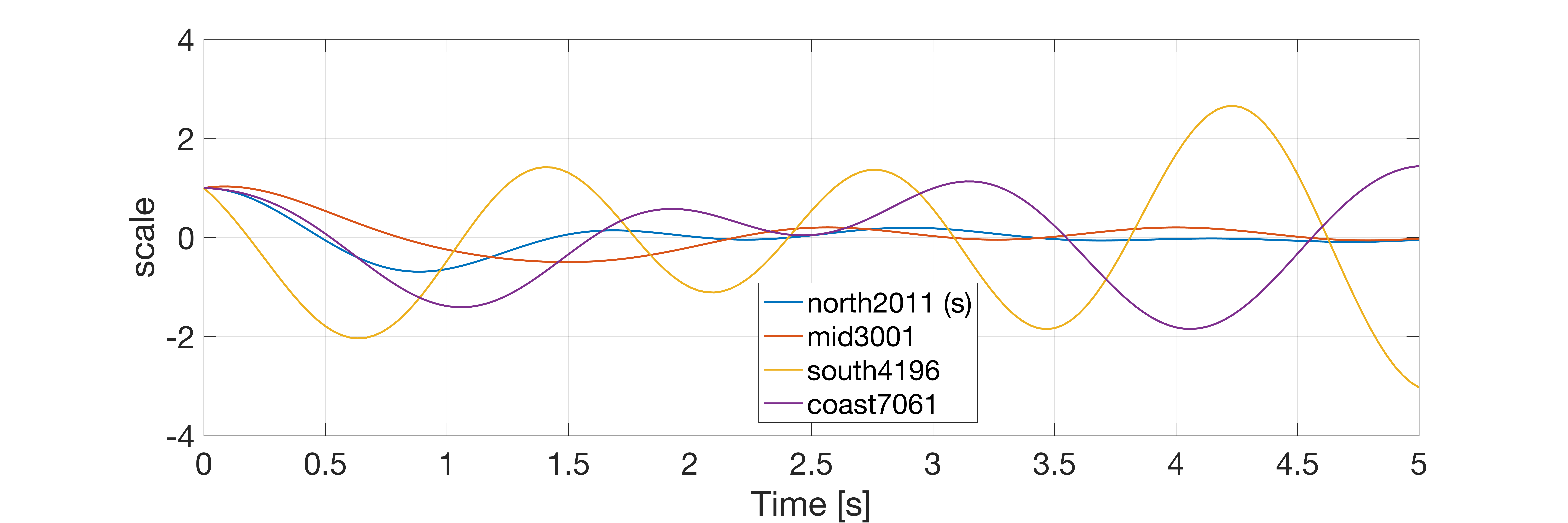}
        \centerline{(c) Synthetic frequency data}
        \caption{Comparison between recovered frequency responses from (a) simulated frequency data; (b) synthetic angle data (with zoom-in view of the nadir points); and (c) synthetic frequency data at four bus locations, with input disturbance at Bus 2011 in the north region.}
        \label{fig:freq_resp}
\end{figure}

\begin{figure}[!t]
  \centering
  \hspace*{-8mm}
  \includegraphics[width=100mm]{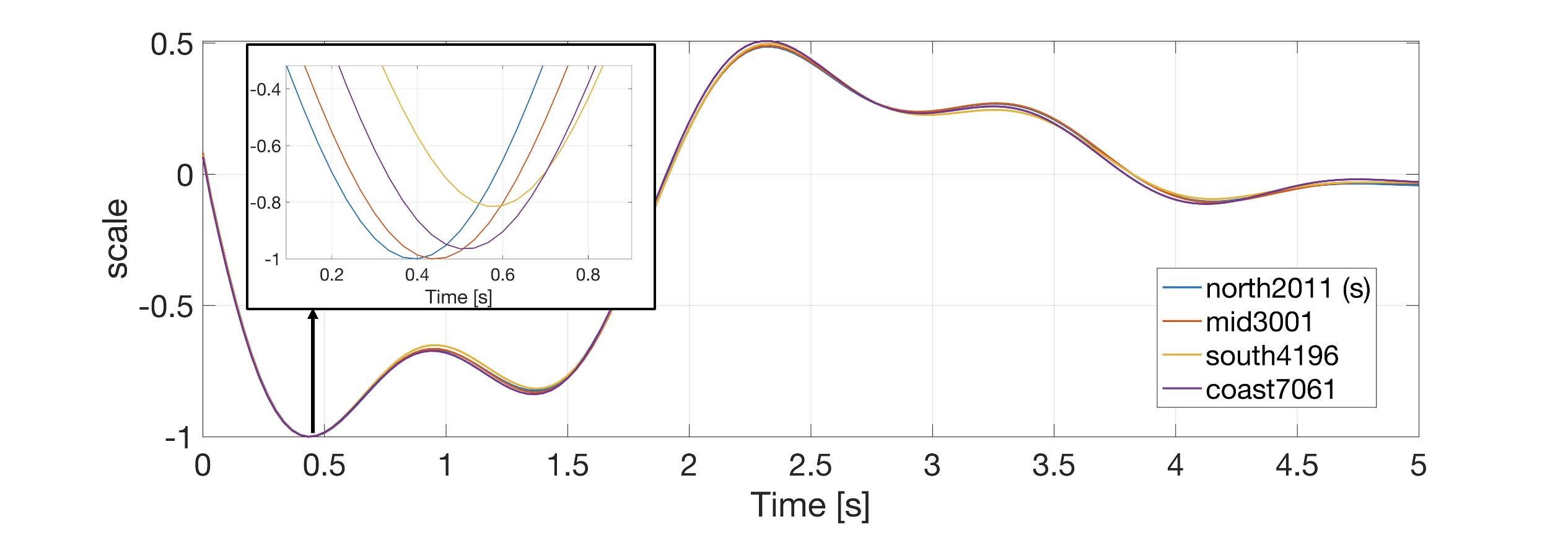} 
  \caption{Recovered angle responses from ambient angle data (with zoom-in view of the nadir points) at four bus locations with input disturbance from Bus 2011 in the north region.}
  \label{fig:ang_resp}
\end{figure}




\section{Conclusions} 
\label{sec:con}
This paper develops a general data-driven framework for recovering small-signal dynamic responses from various types of ambient synchrophasor data. We have proposed a  cross-correlation approach to process ambient data from any two locations of interest. Using the second-order swing dynamics, we have established analytically the equivalence between the proposed cross-correlation and the inertia-based dynamic responses under mild assumptions that hold for large-scale power systems. This equivalence allows to develop a general algorithm to process ambient frequency, angle, and line flow data, that is flexible in both location and type of PMU data. 
The effectiveness of the proposed algorithm under realistic, higher-order, and nonlinear dynamics  has been validated on the WSCC 9-bus case. The case study on the synthetic TI system with realistically generated synchrophasor data further confirms the benefit of using angle data over frequency data, and the applicability to large system implementation. Thus, our proposed data-driven algorithm is attractive in terms of being flexible with the types of synchrophasor data while not requiring any grid modeling information.

Exciting future research directions include the integration with dynamic model estimation for system parameters, and the consideration of high-dimensional data analytics that utilize system-wide data. Furthermore, we are actively exploring new applications for the proposed framework in localizing forced oscillation modes and evaluating advanced control designs.


%

\balance

\section*{References}
\bibliography{IEEEabbrv,ref}

\begin{thebibliography}{10}
\providecommand{\url}[1]{#1}
\csname url@samestyle\endcsname
\providecommand{\newblock}{\relax}
\providecommand{\bibinfo}[2]{#2}
\providecommand{\BIBentrySTDinterwordspacing}{\spaceskip=0pt\relax}
\providecommand{\BIBentryALTinterwordstretchfactor}{4}
\providecommand{\BIBentryALTinterwordspacing}{\spaceskip=\fontdimen2\font plus
\BIBentryALTinterwordstretchfactor\fontdimen3\font minus
  \fontdimen4\font\relax}
\providecommand{\BIBforeignlanguage}[2]{{%
\expandafter\ifx\csname l@#1\endcsname\relax
\typeout{** WARNING: IEEEtran.bst: No hyphenation pattern has been}%
\typeout{** loaded for the language `#1'. Using the pattern for}%
\typeout{** the default language instead.}%
\else
\language=\csname l@#1\endcsname
\fi
#2}}
\providecommand{\BIBdecl}{\relax}
\BIBdecl

\bibitem{kundur1994power}
P.~Kundur, N.~J. Balu, and M.~G. Lauby, \emph{Power System Stability and
  Control}.\hskip 1em plus 0.5em minus 0.4em\relax McGraw-hill, 1994.

\bibitem{kosterev1999model}
D.~N. Kosterev, C.~W. Taylor, and W.~A. Mittelstadt, ``Model validation for the
  august 10, 1996 {WSCC} system outage,'' \emph{{IEEE} Trans. Power Syst.},
  vol.~14, no.~3, pp. 967--979, 1999.

\bibitem{blackout2005}
Y.~V. {Makarov}, V.~I. {Reshetov}, A.~{Stroev}, and I.~{Voropai}, ``Blackout
  prevention in the {United States, Europe, and Russia},'' \emph{Proc. {IEEE}},
  vol.~93, no.~11, pp. 1942--1955, 2005.

\bibitem{adams2012sub_synchronous}
J.~Adams, C.~Carter, and S.-H. Huang, ``{ERCOT} experience with sub-synchronous
  control interaction and proposed remediation,'' in \emph{PES T\&D}, 2012.

\bibitem{nerc_odessa_report2021}
\BIBentryALTinterwordspacing
NERC and T.~R. Staff, ``Odessa disturbance,'' NERC, Tech. Rep., 2021. [Online].
  Available:
  \url{https://www.nerc.com/pa/rrm/ea/Pages/May-June-2021-Odessa-Disturbance.aspx}
\BIBentrySTDinterwordspacing

\bibitem{nerc_report2021}
\BIBentryALTinterwordspacing
S.~M.~W. Group, ``Oscillation analysis monitoring and mitigation,'' {NERC},
  Tech. Rep., 2021. [Online]. Available:
  \url{https://www.nerc.com/comm/RSTC_Reliability_Guidelines}
\BIBentrySTDinterwordspacing

\bibitem{zhang2016dependency}
K.~Zhang, H.~Zhu, and S.~Guo, ``Dependency analysis and improved parameter
  estimation for dynamic composite load modeling,'' \emph{{IEEE} Trans. Power
  Syst.}, vol.~32, no.~4, 2016.

\bibitem{chavan2016identification}
G.~Chavan, M.~Weiss, A.~Chakrabortty, S.~Bhattacharya, A.~Salazar, and F.-H.
  Ashrafi, ``Identification and predictive analysis of a multi-area {WECC}
  power system model using synchrophasors,'' \emph{IEEE Trans.~Smart Grid},
  vol.~8, no.~4, 2016.

\bibitem{pierre1997initial}
J.~W. Pierre, D.~J. Trudnowski, and M.~K. Donnelly, ``Initial results in
  electromechanical mode identification from ambient data,'' \emph{{IEEE}
  Trans. Power Syst.}, vol.~12, no.~3, 1997.

\bibitem{zhou2009electromechanical}
N.~Zhou, Z.~Huang, L.~Dosiek, D.~Trudnowski, and J.~W. Pierre,
  ``Electromechanical mode shape estimation based on transfer function
  identification using {PMU} measurements,'' in \emph{Proc. {IEEE} Power \&
  Energy Society General Meeting}, 2009.

\bibitem{ning2014two}
J.~Ning, S.~A.~N. Sarmadi, and V.~Venkatasubramanian, ``Two-level ambient
  oscillation modal estimation from synchrophasor measurements,'' \emph{{IEEE}
  Trans. Power Syst.}, vol.~30, no.~6, 2014.

\bibitem{nerc_freq_resp}
\BIBentryALTinterwordspacing
T.~Blalock, B.~Cummings, and R.~Bauer, ``Generator governor frequency
  response,'' NERC, Tech. Rep., 2015. [Online]. Available:
  \url{https://www.nerc.com/pa/rrm/Webinars\%20DL/Generator_Governor_Frequency_Response_Webinar_April_2015.pdf}
\BIBentrySTDinterwordspacing

\bibitem{wang2017pmu}
X.~Wang, J.~W. Bialek, and K.~Turitsyn, ``{PMU}-based estimation of dynamic
  state {J}acobian matrix and dynamic system state matrix in ambient
  conditions,'' \emph{{IEEE} Trans. Power Syst.}, vol.~33, 2017.

\bibitem{sheng2020online}
H.~{Sheng} and X.~{Wang}, ``Online measurement-based estimation of dynamic
  system state matrix in ambient conditions,'' \emph{{IEEE} Trans. Smart Grid},
  vol.~11, no.~1, pp. 95--105, 2020.

\bibitem{cui2017inter}
Y.~Cui, L.~Wu, W.~Yu, Y.~Liu, W.~Yao, D.~Zhou, and Y.~Liu, ``Inter-area
  oscillation statistical analysis of the us eastern interconnection,''
  \emph{the Journal of Engr.}, vol. 2017, no.~11, pp. 595--605, 2017.

\bibitem{huynh2018data}
P.~Huynh, H.~Zhu, Q.~Chen, and A.~E. Elbanna, ``Data-driven estimation of
  frequency response from ambient synchrophasor measurements,'' \emph{{IEEE}
  Trans. Power Syst.}, vol.~33, 2018.

\bibitem{pmu_report2020}
A.~Silverstein and J.~Follum, ``High-resolution, time-synchronized grid
  monitoring devices,'' \emph{NASPI Technical Report}, 2020.

\bibitem{jalali2022inferring}
M.~Jalali, V.~Kekatos, S.~Bhela, H.~Zhu, and V.~Centeno, ``Inferring power
  system dynamics from synchrophasor data using {G}aussian processes,''
  \emph{{IEEE} Trans. Power Syst.}, 2022, (early access).

\bibitem{GP4dynamicsCDC21}
M.~Jalali, V.~Kekatos, S.~Bhela, and H.~Zhu, ``Inferring power system frequency
  oscillations using {G}aussian processes,'' in \emph{Proc. {IEEE} Conf. on
  Decision and Control}, Austin, TX, Dec. 2021, pp. 3670--3676.

\bibitem{arthur2000power}
R.~B. Arthur and V.~Vittal, ``Power system analysis,'' \emph{2nd ed., London:
  UK}, pp. 532--538, 2000.

\bibitem{stirzaker1992probability}
D.~Stirzaker and D.~Grimmett, \emph{Probability and random processes}.\hskip
  1em plus 0.5em minus 0.4em\relax Clarendon Press, 1992.

\bibitem{markham2014electromechanical}
P.~N. Markham and Y.~Liu, ``Electromechanical speed map development using
  {FNET/GridEye} frequency measurements,'' in \emph{2014 IEEE PES General
  Meeting}.\hskip 1em plus 0.5em minus 0.4em\relax IEEE, 2014.

\bibitem{birchfield2016grid}
A.~B. Birchfield, T.~Xu, K.~M. Gegner, K.~S. Shetye, and T.~J. Overbye, ``Grid
  structural characteristics as validation criteria for synthetic networks,''
  \emph{{IEEE} Trans. Power Syst.}, vol.~32, no.~4, pp. 3258--3265, 2017.

\bibitem{idehen2020large}
I.~Idehen, W.~Jang, and T.~J. Overbye, ``Large-scale generation and validation
  of synthetic {PMU} data,'' \emph{{IEEE} Trans. Smart Grid}, vol.~11, no.~5,
  2020.

\bibitem{Paganini19}
F.~{Paganini} and E.~{Mallada}, ``Global analysis of synchronization
  performance for power systems: Bridging the theory-practice gap,''
  \emph{{IEEE} Trans. Autom. Contr.}, vol.~65, no.~7, 2020.

\bibitem{osti_1004165}
\BIBentryALTinterwordspacing
P.~Mackin, R.~Daschmans, B.~Williams, B.~Haney, R.~Hung, and J.~Ellis,
  ``Dynamic simulation studies of the frequency response of the three {U.S.}
  interconnections with increased wind generation,'' LBNL, Tech. Rep., 12 2010.
  [Online]. Available: \url{https://www.osti.gov/biblio/1004165}
\BIBentrySTDinterwordspacing

\bibitem{strang2006linear}
G.~Strang, \emph{Linear algebra and its applications}.\hskip 1em plus 0.5em
  minus 0.4em\relax Belmont, CA: Thomson, Brooks/Cole, 2006.

\bibitem{arghir2018grid}
C.~Arghir, T.~Jouini, and F.~D{\"o}rfler, ``Grid-forming control for power
  converters based on matching of synchronous machines,'' \emph{Automatica},
  vol.~95, pp. 273--282, 2018.

\bibitem{poolla2019placement}
B.~K. Poolla, D.~Gro{\ss}, and F.~D{\"o}rfler, ``Placement and implementation
  of grid-forming and grid-following virtual inertia and fast frequency
  response,'' \emph{{IEEE} Trans. Power Syst.}, vol.~34, 2019.

\bibitem{peydayesh2017simplified}
M.~Peydayesh and R.~Baldick, ``Simplified model of {ERCOT} frequency response
  validated and tuned using {PMUs} data,'' \emph{{IEEE} Trans. Smart Grid},
  vol.~9, no.~6, pp. 6666--6673, 2017.

\bibitem{chow2013power}
J.~H. Chow, \emph{Power system coherency and model reduction}.\hskip 1em plus
  0.5em minus 0.4em\relax Springer, 2013, vol.~84.

\bibitem{cotilla2013multi}
E.~Cotilla-Sanchez, P.~D. Hines, C.~Barrows, S.~Blumsack, and M.~Patel,
  ``Multi-attribute partitioning of power networks based on electrical
  distance,'' \emph{{IEEE} Trans. Power Syst.}, vol.~28, no.~4, pp. 4979--4987,
  2013.

\bibitem{milano2005open}
F.~Milano, ``An open source power system analysis toolbox,'' \emph{{IEEE}
  Trans. Power Syst.}, vol.~20, no.~3, pp. 1199--1206, 2005.

\bibitem{meirovitch2001fundamentals}
L.~Meirovitch, \emph{Fundamentals of Vibrations}, ser. McGraw-Hill higher
  education.\hskip 1em plus 0.5em minus 0.4em\relax McGraw-Hill, 2001.

\bibitem{ercot_control2016}
\BIBentryALTinterwordspacing
ERCOT, ``{ERCOT} fundamentals training manual,'' ERCOT, Tech. Rep., 2016.
  [Online]. Available:
  \url{https://www.ercot.com/mp/data-products/services/?id=OPG-160-M}
\BIBentrySTDinterwordspacing

\bibitem{alharbi2020simulation}
A.~Alharbi and S.~You, ``A simulation-based education approach for the
  electromagnetic and electromechanical transient waves in power systems,''
  \emph{arXiv preprint arXiv:2010.12379}, 2020.

\end{thebibliography}
\bibliographystyle{IEEEtran}

\section*{Acknowledgments}
This work has been supported by the US National Science Foundation (NSF) awards ECCS-1802319, 1751085, 2150571, and 2150596. The authors would like to thank Dr. Ikponmwosa Idehen for providing the synthetic Texas system data.

\end{document}